\documentclass[letterpaper]{article} 
\usepackage{aaai19}  
\usepackage{times}  
\usepackage{helvet}  
\usepackage{courier}  
\usepackage{url}  
\usepackage{graphicx}  
\nocopyright
\usepackage{mathtools,amssymb,amsthm}
\usepackage[T1]{fontenc}
\newcommand{\Z}{\mathbb{Z}}
\newcommand{\cur}{\text{cur}}
\newcommand{\Ssk}{S_{\text{sk}}}
\newcommand{\Ssy}{S_{\text{sy}}}
\newcommand{\Squ}{S_{\text{q}}}
\newtheorem{definition}{Definition}
\newtheorem{theorem}{Theorem}
\newtheorem{corollary}[theorem]{Corollary}
\newtheorem{lemma}[theorem]{Lemma}
\DeclareMathOperator{\PSD}{PSD}
\DeclareMathOperator{\row}{sum}
\DeclareMathOperator{\Aut}{Aut}

\DeclareMathOperator{\PAF}{PAF}
\DeclarePairedDelimiter{\abs}{\lvert}{\rvert}

\DeclarePairedDelimiter{\floor}{\lfloor}{\rfloor}

\makeatletter
\usepackage{hyperref}
\renewcommand{\@BIBLABEL}{\@emptybiblabel}
\newcommand{\@emptybiblabel}[1]{}
\hypersetup{breaklinks,hidelinks}
\makeatother
\begin{document}
\title{A SAT+CAS Approach to Finding Good Matrices: \\New Examples and Counterexamples}
\author{Curtis Bright \\ University of Waterloo \\
\And Dragomir \v{Z}.\ \DJ okovi\'c \\ University of Waterloo \\
\And Ilias Kotsireas \\ Wilfrid Laurier University \\
\And Vijay Ganesh \\ University of Waterloo}
\maketitle
\begin{abstract}
We enumerate all circulant good matrices with odd orders divisible by $3$ up to order~$70$.
As a consequence of this we find a previously
overlooked set of good matrices of order~$27$ and
a new set of good matrices of order $57$.
We also find that circulant good matrices do not exist in
the orders $51$, $63$, and~$69$, thereby
finding three new counterexamples to the conjecture that such matrices
exist in all odd orders.
Additionally, we prove a new relationship
between the entries of good matrices and exploit this relationship in
our enumeration algorithm.  Our method applies the SAT+CAS paradigm
of combining computer algebra functionality with modern SAT solvers to efficiently search
large spaces which are specified by both algebraic and logical constraints.
\end{abstract}

\section{Introduction}

In 2002, circulant ``good'' matrices were searched for in all odd orders $n$ with $n<40$~\cite{georgiou2002good}
and were found to exist in all such orders.
This gave evidence to the conjecture that such matrices actually exist in \emph{all} odd orders.
The mathematician George
Szekeres 
studied circulant good matrices in terms of an equivalent type of object which he termed
$E$-sequences and in his classic paper~\cite{szekeres1988note} he
mentioned this hypothesis as being worthwhile to study using computers:
\begin{quote}
\ldots it is conceivable that $E$-sequences exist for all $n=2m+1$,
$m\geq1$ and it is worth testing this hypothesis at least for those
orders which are accessible to present day computers\ldots
\end{quote}
Unfortunately, this conjecture was recently shown to be false
in~\cite{djokovic2018goethals} where it was shown that
the orders $n=41$, $47$, and $49$ are counterexamples,
i.e., no circulant good matrices of orders $41$, $47$, and $49$ exist.

In this paper we find the additional three larger counterexamples $n=51$, $63$, and $69$.
In addition, we verify the claim of \DJ okovi\'c and Kotsireas
that there are exactly four inequivalent sets of circulant good matrices of order $45$
and determine that there is a single set of inequivalent circulant good matrices of order $57$.
We also find a previously undiscovered set of good matrices of order $27$
demonstrating that Szekeres' claimed exhaustive search~\cite{szekeres1988note} was incomplete.

The matrices that are now known as good matrices were first used in
the PhD thesis of Jennifer Seberry Wallis~\cite{wallis1970combinatorial};
we give their formal definition in Section~\ref{sec:background}.
She gave a construction using good matrices that enabled
the construction of \emph{Hadamard matrices}%
---square matrices with $\pm1$ entries and whose rows are pairwise orthogonal.
In her thesis she gave examples of good matrices in all odd orders up to $15$
as well as order $19$ and shortly later gave an example in order $23$~\cite{wallis1971skew}.
Subsequently, \cite{hunt1972skew} ran a complete search in the odd orders up to $21$
and gave examples in order~$25$,
\cite{szekeres1988note} ran a complete search up to order~$31$,
\cite{djokovic1993good} gave examples in orders~$33$ and~$35$,
\cite{georgiou2002good} ran a complete search up to order~$39$
and most recently \cite{djokovic2018goethals}
ran a complete search up to order~$49$.
The result of these searches have shown that good matrices exist
in all odd orders up to $39$ as well as $43$ and $45$.
Additionally, \cite{djokovic1993good}
gave a construction showing they exist in order $127$.

In this paper we extend the search to all odd orders divisible by $3$
up to order $69$.
Note that increasing $n$ by $2$ increases the size of the search space
(at least to a first approximation)
by a factor of $16$,
meaning that the size of the search space 
in order $69$ is an enormous $2^{40}$
times larger than the search space in order $49$, the largest
order which has previously been exhaustively searched.
There are two primary reasons our approach is able to scale to such 
orders:
\begin{enumerate}
\item We use a number of theoretical mathematical results
to cut down the search space,
including filtering theorems, recently proven compression
theorems, and a new product theorem given in this paper.
\item We use state-of-the-art programmatic SAT solvers,
computer algebra systems, and mathematical libraries to very efficiently
search spaces specified by both logical and algebraic constraints.
\end{enumerate}

This approach of coupling computer algebra systems and SAT solvers
was first proposed in 2015 at the conferences CADE
and ISSAC~\cite{zulkoskimathcheck,abraham2015building} and has since appeared
in papers at the conferences IJCAI, ISSAC,
and the journal of automated
reasoning~\cite{ZulkoskiGC16,bright2018enumeration,ZulkoskiBHKCG17}
and is the aim of the SC$^2$ project~\cite{sc2}.
In particular, the SAT+CAS method in this paper is an adaption
of one presented at AAAI~\cite{bright2018sat+}.

Furthermore, we employ a SAT solver which can learn clauses
\emph{programmatically}, through a piece of code compiled with
the SAT solver and which makes calls to an external mathematical
library.  This approach permits one to efficiently encode constraints
that would be much too cumbersome to encode with native SAT clauses,
as we will see in Section~\ref{sec:sat+cas}.
For example, it is much faster and easier to compute a discrete Fourier
transform using a numerical library or computer algebra system than it
would be to encode a DFT circuit using logical clauses, the format
normally accepted by SAT solvers.

A detailed description of our enumeration algorithm
will be given in Section~\ref{sec:method}, followed by our
results in Section~\ref{sec:results}.  In particular, we
explicitly give two new examples of good matrices,
one in order $27$ which has gone undetected since 1988
and one in order $57$, an order which had previously been
out of reach of exhaustive search algorithms; the largest order
searched prior to our work was $49$.
We also provide a table enumerating the number of
sets of good matrices in all odd orders
divisible by $3$ up to $69$.  These counts are given
up to an equivalence which is described, along with
the required background on good matrices, in
Section~\ref{sec:background}.

\section{Background}\label{sec:background}

In this section we define good matrices and review the
properties and theorems of good matrices that we use
in our enumeration algorithm.

To begin, we recall the
definitions of symmetric, skew, and circulant matrices.
Let $X$ be a square matrix of order~$n$ with entries
$x_{i,j}$ with $0\leq i,j<n$.
We say that $X$ is \emph{symmetric}
if all entries satisfy $x_{i,j}=x_{j,i}$,
that $X$ is \emph{skew}
if diagonal entries are $1$ and non-diagonal entries satisfy $x_{i,j}=-x_{j,i}$,
and that $X$ is \emph{circulant}
if all entries satisfy $x_{i,j}=x_{i-1,j-1}$
with indices taken mod~$n$ as necessary.

\subsection{Good matrices}

\begin{definition}\label{def:good}
Four matrices $A$, $B$, $C$, $D\in\{\pm 1\}^{n\times n}$ are known as \emph{good matrices}
if they have the following properties.
\begin{enumerate}
\item[(a)] They are pairwise amicable ($XY^T$ is symmetric for $X,Y\in\{A,B,C,D\}$).
\item[(b)] $A$ is skew and $B$, $C$, $D$ are symmetric.
\item[(c)] $AA^T + B^2 + C^2 + D^2$ is the identity matrix scaled by $4n$.
\end{enumerate}
\end{definition}

The primary reason that good matrices have 
attracted sustained interest for almost 50 years
is because of the following theorem
of \cite{wallis1970combinatorial} which
provides a construction for
skew Hadamard matrices.

\begin{theorem}
Let $A$, $B$, $C$, $D$ be good matrices of order $n$.  Then
\[ \begin{bmatrix}
A & B & C & D \\
-B & A & D & -C \\
-C & -D & A & B \\
-D & C & -B & A
\end{bmatrix} \]
is a skew Hadamard matrix of order $4n$.
\end{theorem}

Skew Hadamard matrices are conjectured to exist in all orders $4n$ with $n\geq1$
and much effort has gone into constructing them in as many orders
as possible; the current smallest unknown order is $n=69$
and the previous smallest unknown order $n=47$ was solved ten years ago
by~\cite{djokovic2007skew}.

From a computational perspective the issue with good matrices
is that the search space is too large to effectively search.
The standard remedy for this is
to focus on \emph{circulant} matrices for which
the search can be performed more effectively.
For example, in~\cite{gene2015} good matrices are defined
as in Definition~\ref{def:good} except that condition~(a) is replaced
with the strictly stronger condition
\textit{\begin{enumerate}
\item[(a$'$)] $A$, $B$, $C$, $D$ are circulant.
\end{enumerate}}
Technically matrices which satisfy conditions (a$'$), (b), and~(c)
are not good matrices but
one can recover good matrices from them simply by reversing the order
of the rows of $B$, $C$, and~$D$ (thereby making them amicable with $A$).
Since using~(a$'$) is convenient to state our results
we will employ it for the remainder of this paper.

We often consider good matrices in terms of their first rows
$(a_0,\dotsc,a_{n-1})$, $(b_0,\dotsc,b_{n-1})$, $(c_0,\dotsc,c_{n-1})$,
and $(d_0,\dotsc,d_{n-1})$ which we call their \emph{defining rows}.
Since $A$ is skew and $B$, $C$, and~$D$ are symmetric, the defining rows
of a set of good matrices must satisfy $a_0=1$ and
\[ a_i=-a_{n-i},\quad b_i=b_{n-i},\quad c_i=c_{n-i},\quad d_i=d_{n-i} \label{eq:sym} \]
for $1\leq i<n/2$.  Furthermore, we often
use $A$, $B$, $C$, and~$D$ to denote the defining rows themselves,
in which case we say that $A$ is a skew sequence and that
$B$, $C$, and~$D$ are symmetric sequences.

\subsection{Equivalence operations}\label{sec:equiv}

The following three invertible operations are well-known and
can be applied to a set of
good matrices $A$, $B$, $C$, $D$ to produce another set of good matrices
we consider \emph{equivalent} to $A$, $B$, $C$, $D$.

\begin{enumerate}
\item Reorder $B$, $C$, $D$ in any way.
\item Negate any of $B$, $C$, $D$.
\item Apply an automorphism of $\Z_n=\{0,\dotsc,n-1\}$ to the indices of the defining rows
of each of $A$, $B$, $C$, and~$D$.
\end{enumerate}

For example, the negation equivalence operation implies that any set of good
matrices is equivalent to one with $b_0=c_0=d_0=1$.

\subsection{Filtering theorems}

The search space for good matrices of order~$n$ is enormous.
Even if one takes into account the symmetry properties of
the defining rows
and fixes the first entries to be positive, each defining row
contains $\floor{n/2}$ unspecified entries.  Since each
unspecified entry can be one of two values, there are a total of $2^{4\floor{n/2}}=4^{n-1}$
possibilities in the search space.

In order to cut down the size of the search space so that we can run
exhaustive searches in larger orders we need filtering theorems.
One of the most powerful filtering theorems is based on the
following which gives an alternative characterization of good matrices~\cite{djokovic2015compression}.

\begin{theorem}\label{thm:psd}
Let $A$, $B$, $C$, $D$ be sequences of length $n$
with $A$ skew and $B$, $C$, $D$ symmetric.
Then $A$, $B$, $C$, $D$ are the defining rows of a set
of good matrices if and only if
for all integers $k$ we have 
\[ \PSD_A(k) + \PSD_B(k) + \PSD_C(k) + \PSD_D(k) = 4n . \]
Here $\PSD_X(k) \coloneqq \abs[\big]{\sum_{j=0}^{n-1}x_j\omega^{jk}}^2$
with $\omega\coloneqq\exp(2\pi i/n)$ is
the power spectral density of $X=[x_0,\dotsc,x_{n-1}]$.
\end{theorem}

This is important because of the following well-known corollary:

\begin{corollary}\label{cor:psd}
Let $A$, $B$, $C$, $D$ be the defining rows of a set of good matrices
of order~$n$.
If\/ $S$ is a subset of\/ $\{A,B,C,D\}$ then
\[ \sum_{X\in S} \PSD_X(k) \leq 4n \]
for all integers $k$.
\end{corollary}

For example, this corollary says that if $\PSD_A(k) > 4n$ for some
$k$ then $A$ cannot be part of a set of good matrices.
In other words, we can filter $A$ and ignore it in our enumeration method.
A second well-known corollary of Theorem~\ref{thm:psd} comes from setting $k=0$
in the statement of the theorem:

\begin{corollary}\label{cor:dio}
Let $A$, $B$, $C$, $D$ be the defining rows of a set of good matrices
of odd order~$n$.
Then
\[ \row(B)^2 + \row(C)^2 + \row(D)^2 = 4n-1 \]
where $\row(X)$ denotes the sum of entries in $X$.
\end{corollary}

Since $\row(X)$ is an integer when $X$ has integer entries
this corollary tells us that every set of good matrices gives
a decomposition of $4n-1$ into a sum of three integer squares
which is a rather restrictive condition.  For example,
when~$n$ is~$69$ there are only three ways to write $4n-1$
as a sum of exactly three squares of nonnegative integers, namely,
\[ 1^2+7^2+15^2,\quad 5^2+5^2+15^2,\quad\text{and}\quad 5^2+9^2+13^2 . \]
One must also consider decompositions with
squares of negative integers
but the negative integers
can be found with the following simple (often unstated) lemma.
It gives a criterion to determine the signs
of $\row(B)$, $\row(C)$, and $\row(D)$.
\begin{lemma}\label{lem:sign}
Let $(a_0,\dotsc,a_{n-1})$, $(b_0,\dotsc,b_{n-1})$, $(c_0,\dotsc,c_{n-1})$, $(d_0,\dotsc,d_{n-1})$ be the defining rows
of a set of good matrices $A$, $B$, $C$, $D$ of odd order $n$ with $b_0=c_0=d_0=1$.  Then
\[ \row(B) \equiv \row(C) \equiv \row(D) \equiv n \pmod{4} . \]
\end{lemma}

\subsection{Compression}

Although the filtering theorems greatly decrease the
number of possibilities for the defining rows of good matrices there
are still typically a large number of possibilities which are not
filtered by Corollaries~\ref{cor:psd} and~\ref{cor:dio}.
An effective way of shrinking the search space even farther is to apply
properties that ``compressions'' of good matrices must satisfy.
Such theorems are only applicable in composite orders; we will assume that
the order~$n$ is a multiple of $3$.
To our knowledge this is the first time compression has
been used in the search for good matrices and it allows us
to search larger orders than other methods.  For example,
\cite{djokovic2018goethals} use the aforementioned filtering
corollaries but do not apply compression and are only able to search orders
up to~$49$.

The \emph{3-compression} of a sequence $X=(x_0,\dotsc,x_{n-1})$
of length $n$ is the sequence $X'$ of length $m\coloneqq n/3$ whose $k$th
entry is $x_k+x_{k+m}+x_{k+2m}$.
Since we are concerned with the case where the entries of $X$
are $\pm1$ the entries of $X'$ are either $\pm1$
or $\pm3$.
Compression is useful because
\cite{djokovic2015compression} showed that
Theorem~\ref{thm:psd} holds for compressions of the defining
rows of good matrices.

\begin{theorem}\label{thm:compress}
Let $A$, $B$, $C$, $D$ be the defining rows of a set of good matrices
of order~$n$ and let $A'$, $B'$, $C'$, $D'$ be compressions of
those rows.
Then for all integers $k$ we have 
\[ \PSD_{A'}(k) + \PSD_{B'}(k) + \PSD_{C'}(k) + \PSD_{D'}(k) = 4n . \]
\end{theorem}

\subsection{A new product theorem}

Our study of good matrices
uncovered the following relationship that the entries
of good matrices must satisfy.  This is an analog of the product
theorem that was proven by John Williamson
for Williamson matrices~\cite{williamson1944hadamard}.

\begin{theorem}\label{thm:product}
Let $(a_0,\dotsc,a_{n-1})$, $(b_0,\dotsc,b_{n-1})$, $(c_0,\dotsc,c_{n-1})$, $(d_0,\dotsc,d_{n-1})$ be the defining rows
of a set of good matrices $A$, $B$, $C$, $D$ of odd order $n$.  Then $a_k b_k c_k d_k = -a_{2k\bmod n}b_0c_0d_0$ for all\/ $1\leq k<n$.
\end{theorem}

Let $\bar{x}\coloneqq(1-x)/2$, i.e., the mapping $x\mapsto\bar{x}$ is the group
isomorphism between $\Z_4^*=\{\pm1\}$ and $\Z_2=\{0,1\}$.
Our enumeration method uses Theorem~\ref{thm:product} in the following form.

\begin{corollary}\label{cor:parity}
Let $(a_0,\dotsc,a_{n-1})$, $(b_0,\dotsc,b_{n-1})$, $(c_0,\dotsc,c_{n-1})$, $(d_0,\dotsc,d_{n-1})$ be the defining rows
of a set of good matrices $A$, $B$, $C$, $D$ of odd order $n$ with
$b_0=c_0=d_0=1$.
Then for $1\leq k<n/2$ we have in\/ $\Z_2$
\[ \bar{a}_k + \bar{a}_{2k} + \bar{b}_k + \bar{c}_k + \bar{d}_k = 1 \]
and when $n=3m$ we have $\bar{b}_{m}+\bar{c}_{m}+\bar{d}_{m}=1$.
\end{corollary}

We give proofs of Lemma~\ref{lem:sign}, Theorem~\ref{thm:product},
and Corollaries~\ref{cor:psd}, \ref{cor:dio}, and~\ref{cor:parity}
in the appendix.

\section{The SAT+CAS paradigm}\label{sec:sat+cas}

The SAT+CAS paradigm is a new method of solving certain kinds of computational
problems which harnesses the power of two fields of computer science:
satisfiability checking and symbolic computation.
The paradigm is particularly useful for problems that have
a significant search component and a rich mathematical
component.

Briefly, this is due to the fact that modern SAT solvers
have extremely efficient search procedures
but lack the domain knowledge to effectively deal with rich mathematical
concepts.  On the other hand, computer algebra systems like
\textsc{Maple}, \textsc{Mathematica}, and \textsc{SageMath}
as well as special-purpose numerical libraries
can very effectively deal with rich mathematics.

The fact that satisfiability checking and symbolic computation had great potential synergy
was first pointed out by Erika \'Abrah\'am in an invited
talk at the conference ISSAC in 2015~\cite{abraham2015building}.
At almost the same time this synergy was demonstrated
by the system \textsc{MathCheck}
presented at the conference CADE~\cite{zulkoskimathcheck}.
The system \textsc{MathCheck} coupled a SAT solver with a computer
algebra system and solved open cases
of two conjectures in graph theory and was later extended to
solve open cases in combinatorial conjectures~\cite{ZulkoskiBHKCG17}.
Since then, the $\text{SC}^2$ project~\cite{sc2}
has organized an annual workshop on this topic for the last three years.

To illustrate the usefulness of the SAT+CAS paradigm in the search for good matrices
specifically, consider Corollaries~\ref{cor:psd}, \ref{cor:dio}, and~\ref{cor:parity}
from Section~\ref{sec:background}.  These corollaries all state
various properties that good matrices must satisfy
but each property arises from a different branch of mathematics.
Consequently, each property is best dealt with using systems which
have been optimized to deal with that particular branch:
\begin{itemize}
\item Corollary~\ref{cor:psd} is a statement in Fourier analysis
and is best checked using a system which has been
designed to efficiently compute Fourier transforms.
\item Corollary~\ref{cor:dio} (and Lemma~\ref{lem:sign}) are
number theoretic statements and possibilities
for the values $\row(B)$, $\row(C)$, and $\row(D)$
can be found by using a system which has been designed
to solve quadratic Diophantine systems.
\item Corollary~\ref{cor:parity} is a statement using
arithmetic in $\Z_2$ and can efficiently be encoded directly
in a SAT instance.
\end{itemize}

\subsection{The programmatic SAT paradigm}

A programmatic SAT solver, introduced in~\cite{ganesh2012lynx},
is a new kind of SAT solver which allows solving satisfiability
problems with more expressiveness than the kinds of
problems that can be solved with a standard SAT solver.
Briefly, a programmatic SAT solver will use
some custom code to examine the current partial assignment and
if it cannot be extended into
a full assignment (based on some known theorem from the 
problem domain) a clause will be learned encoding that fact.
Because one can think of ``calling'' the
programmatic SAT solver with the custom code this
code is known as a ``callback'' function.

Programmatic SAT solvers are conceptually similar
to SMT (SAT modulo theories) solvers and in some ways could
be considered a ``poor man's SMT''.
However, they also differ from SMT solvers in some key ways:

\subsubsection{Simplicity.}

SMT solvers are necessarily more
complicated than SAT solvers because they solve problems
in first-order logic instead of problems in propositional
logic.
This increase in complexity is justifiable for many problems
but some problems contain mostly just propositional logic with a handful of
``theory lemmas'' that 
are cumbersome to state in propositional logic but
can be given to a programmatic SAT solver
as necessary to help guide its search.

\subsubsection{Flexibility.}

Typically SMT solvers only support
theories which are fixed in advance;
for example, the SMT-LIB standard~\cite{BarFT-SMTLIB}
specifies a fixed number of theories and a common input and
output language for those theories.  However, not all
problems can naturally be expressed in those theories;
e.g., the current SMT-LIB standard does not include
theories involving Fourier transforms (or transcendental
functions or complex numbers).
In contrast, a programmatic SAT solver can learn
theory lemmas derived using Fourier transforms
so long as an appropriate CAS or library can be called.

\subsubsection{Tailored solving.}

Programmatic SAT solvers
can be tailored to solve specific
classes of problems (in addition to using the same state-of-the-art
techniques which make modern SAT solving so efficient).
For example, in this paper we develop a SAT solver
tailored to searching for good matrices.  In principle SMT
solvers could work in the same way but we are not aware of
any SMT solvers that offer a programmatic interface like this.

\subsection{A programmatic SAT encoding of good matrices}

We now describe how our programmatic SAT solver encoding
of good matrices works in detail.
We encode the entries of the defining rows of a set of good
matrices $A$, $B$, $C$, $D$ of order $n$ using $4n$ Boolean variables
where we let true values denote~$1$ and false values denote~$-1$.
By abuse of notation we use the same names for the Boolean
variables and the $\pm1$ entries that they represent but it
will be clear from the context if we are referring
to a Boolean or integer value.

In fact, using the symmetry constraints from Section~\ref{sec:background} we only
need to define the $2(n+1)$ variables $a_i$, $b_i$, $c_i$, $d_i$
that have indices $i$ with $0\leq i<n/2$.  In what follows we implicitly use this
whenever necessary; i.e., any variables with indices larger than $n/2$
are used for clarity only.

A programmatic SAT solver includes a callback function
that is run whenever the SAT solver's usual conflict analysis
fails to find a conflict.
The callback function examines the current partial assignment
and learns clauses that block the current assignment
(and extensions of the current assignment) from the search in the future.

If $x$ is a variable that appears in the current partial
assignment we let $x^{\cur}$ denote either the literal $x$
or $\lnot x$, whichever is true under the current assignment.
The callback function used in our encoding
of good matrices
works by encoding the property given in Corollary~\ref{cor:psd}
and its steps are described in detail below.

\begin{enumerate}
\item Let $S$ be a list containing the sequences
$(a_0,\dotsc,a_{n-1})$, $(b_0,\dotsc,b_{n-1})$, $(c_0,\dotsc,c_{n-1})$, $(d_0,\dotsc,d_{n-1})$
that have all their variables currently assigned.
\item If $S$ is empty then return control to the SAT solver.
\item Let $X$ be the first sequence in $S$ and compute the
power spectral density of $X$.
If $\PSD_X(k)>4n$ for some $k$ then learn a clause blocking the sequence $X$ from
the search in the future.  Explicitly, learn the clause
\[ \lnot x_0^{\cur}\lor\lnot x_1^{\cur}\lor\dotsb\lor\lnot x_{n-1}^{\cur} \]
and return control to the SAT solver.
\item If $S$ contains at least two sequences then let $Y$ be the second
sequence in $S$ and compute the power spectral density of $Y$.
If $\PSD_X(k)+\PSD_Y(k)>4n$ for some $k$ then learn a clause blocking the
sequences $X$ and $Y$ from occurring together again in the future.  Explicitly, learn the clause
\[ \lnot x_0^{\cur}\lor\dotsb\lor\lnot x_{n-1}^{\cur}\lor\lnot y_0^{\cur}\lor\dotsb\lor\lnot y_{n-1}^{\cur} \]
and return control to the SAT solver.
\item If $S$ contains at least three sequences then let $Z$ be the third sequence in $S$
and compute the power spectral density of $Z$.  If $\PSD_X(k)+\PSD_Y(k)+\PSD_Z(k)>4n$ for some $k$ then learn a clause blocking the
sequences $X$, $Y$, and $Z$ from occurring together again in the future and return control to the SAT solver.
\item If $S$ contains four sequences then we know every entry of $A$, $B$, $C$, and~$D$.
If the relationship in Theorem~\ref{thm:psd} holds then record $A$, $B$, $C$, and~$D$ as rows
that define good matrices.
Whether or not the sequences define good matrices learn a clause
blocking these sequences from occurring together again in the future.
\end{enumerate}

If only a single instance of good matrices is desired then one can
stop searching in Step~6 if a satisfying $A$, $B$, $C$, $D$ is found.
However, we learn a blocking clause and continue the search 
because we want to provide an exhaustive search.

The programmatic approach was essential to our algorithm because it allowed
us to easily and efficiently apply Corollary~\ref{cor:psd} which would
otherwise be difficult to apply in a SAT instance.
Corollary~\ref{cor:psd} was also extremely effective;
the SAT solver was usually able to learn a clause in Step~3
when enough variables were assigned.
This enormously cut down the search
space---the SAT solver could often learn conflicts with
just~$n$ variables instead of conflicts with all $4n$ variables
and in practice this appeared to give us an exponential speedup in $n$.

However, using this method alone was not sufficient for us
to derive results beyond order $21$
and a more efficient enumeration method (using the
programmatic encoding as a subroutine)
was required to scale to order $69$.
We describe this method in the next section.

\section{Enumeration algorithm}\label{sec:method}

Although SAT solvers are useful tools for searching large spaces
they have their limits and the main reason that we were not
able to derive any new results only using SAT solvers is 
because the search space is too large, even using the
programmatic filtering encoding.

One way that the performance of SAT solvers can be improved
on instances with large search spaces is by splitting the
instance into many smaller instances of approximately
equal difficulty.   The cube-and-conquer paradigm~\cite{heule2018cube}
does this and has achieved some impressive successes such as solving the
Boolean Pythagorean triples problem~\cite{heule2016solving} and determining the
fifth Schur number~\cite{heule2018schur}.

Our enumeration algorithm will use a similar paradigm to
cube-and-conquer except that instead of splitting the search
space into \emph{cubes} we split the search space into
\emph{compressions}.  Not only does splitting via
compressions increase the performance of the SAT solver
it also allows for easy parallelization because
instances generated in this way are independent.
Additionally, we use Theorem~\ref{thm:compress}
to apply filtering to
the compressions which also has the effect of
further decreasing the size of the search space.

We now describe our algorithm in four steps
followed by some postprocessing.  Given an odd order~$n$
divisible by~$3$ our algorithm enumerates
all inequivalent circulant good matrices of order~$n$
with positive first entries.
Each step is designed to search through a larger subspace
than the previous step: step~1 finds every possible
rowsum of a good matrix, step~2 finds every possible
$3$-compressed row of a good matrix, step~3 finds every
possible $3$-compressed quadruple of good matrices,
and step~4 finds all good matrices up to equivalence.

\subsection{Step 1: Find possible rowsums}

If $A$, $B$, $C$, $D$ are good matrices then
from Corollary~\ref{cor:dio} and Lemma~\ref{lem:sign}
the rowsums of the matrices $B$, $C$, and~$D$
must satisfy the quadratic Diophantine system
\begin{gather*}
x^2 + y^2 + z^2 = 4n - 1 \\
x \equiv y \equiv z \equiv n \pmod{4}
\end{gather*}
where $x$, $y$, and $z$ are the rowsums of $B$, $C$, and~$D$.

Many computer algebra systems contain number theoretic functions
that allow this Diophantine system to be easily solved.
For example, the \textsc{Mathematica} function
\texttt{PowersRepresentations} or the
\textsc{Maple} script \textsc{nsoks}~\cite{nsoks}.

\subsection{Step 2: Find possible compressed rows}

In this step we will generate sets $\Ssk$ and $\Ssy$ where
$\Ssk$ will contain all possible $3$-compressions of skew
defining rows and $\Ssy$ will contain
all possible $3$-compressions of symmetric defining rows.
These sets will be generated by an essentially
brute-force enumeration
through all skew or symmetric
$\{\pm1\}$-sequences of length $n$, using Corollaries~\ref{cor:psd}
and~\ref{cor:dio} to filter sequences which
cannot be a defining row.

Let $d\coloneqq\floor{n/2}$.  In detail, we do the following:

\begin{enumerate}
\item For all $X\in\{\pm1\}^{d}$:
\begin{enumerate}
\item Let $A$ be the skew sequence whose first entry is~$1$
and next $d$ entries are given by $X$, i.e.,
$A\coloneqq(1,x_0,\dotsc,x_{d-1},-x_{d-1},\dotsc,-x_0)$.
\item If $\PSD_A(k)\leq4n$ for all $k$ then
add the $3$-compression of $A$ to $\Ssk$.
\item Let $B$ be the symmetric sequence whose first entry is~$1$
and next $d$ entries are given by $X$, i.e.,
$B\coloneqq(1,x_0,\dotsc,x_{d-1},x_{d-1},\dotsc,x_0)$.
\item If $\PSD_B(k)\leq4n$ for all $k$ and $\row(B)$ is one
of the solutions $x$, $y$, $z$ from step~1 then
add the $3$-compression of $B$ to $\Ssy$.
\end{enumerate}
\end{enumerate}

\subsection{Step 3: Find possible compressed quadruples}

In step~2 we generated all possible 3-compressions
of the defining rows of good matrices.  However,
most quadruples of possible sequences found in
step~2 will violate the relationship
in Theorem~\ref{thm:compress}
and can therefore be filtered.
In this step we identify all quadruples of 3-compressions
which can't be filtered
and store them in a set $\Squ$.
We do this by examining all $(A',B',C',D')$ in $\Ssk\times\Ssy^3$
and adding the quadruples that satisfy the relationship in Theorem~\ref{thm:compress}
to $\Squ$.

\subsection{Step 4: Uncompress via programmatic SAT}

In step~3 we generated all $3$-compressed quadruples $(A',B',C',D')$
that satisfy the relationship in Theorem~\ref{thm:compress}.  Given such a quadruple we
need to determine if it is possible to \emph{uncompress} the quadruple.
In other words, we want to find all quadruples $(A,B,C,D)$ of defining
rows of good matrices (if any) whose $3$-compression is $(A',B',C',D')$.
To do this we formulate the uncompression problem as a programmatic SAT
problem.

We already discussed our programmatic encoding of the constraint
saying that $A$, $B$, $C$, $D$ define rows of good matrices
of order $n=3m$.
The remaining constraints can easily be expressed in the standard format
accepted by SAT solvers (conjunctive normal form).
We now
discuss how to encode the constraint that they are $3$-compressions
of $A'$, $B'$, $C'$, and~$D'$.
(For simplicity we only consider the compression constraints of
$A'=(a'_0,\dotsc,a'_{m-1})$
but the other constraints are analogous.)

Consider the entry $a'_k=a_k+a_{k+m}+a_{k+2m}$.
There are four possible cases:

\begin{enumerate}
\item $a'_k=3$.  In this case $a_k=a_{k+m}=a_{k+2m}=1$ which
is encoded as the three unit clauses
\[ a_k, \quad a_{k+m}, \quad a_{k+2m} . \]
\item $a'_k=1$.  In this case two of $\{a_k,a_{k+m},a_{k+2m}\}$
are true and one is false.  This is encoded as the four clauses
\begin{gather*}
\lnot a_k \lor \lnot a_{k+m} \lor \lnot a_{k+2m} , \\
a_k \lor a_{k+m} , \quad
a_k \lor a_{k+2m} , \quad
a_{k+m} \lor a_{k+2m} .
\end{gather*}
\item $a'_k=-1$.  In this case one of $\{a_k,a_{k+m},a_{k+2m}\}$
is true and two are false.  This is encoded as the four clauses
\begin{gather*}
a_k \lor a_{k+m} \lor a_{k+2m} , \\
\lnot a_k \lor \lnot a_{k+m} , \quad\!\!
\lnot a_k \lor \lnot a_{k+2m} , \quad\!\!
\lnot a_{k+m} \lor \lnot a_{k+2m} .
\end{gather*}
\item $a'_k=-3$.  In this case $a_k=a_{k+m}=a_{k+2m}=-1$ which
is encoded as the three unit clauses
\[ \lnot a_k, \quad\!\! \lnot a_{k+m}, \quad\!\! \lnot a_{k+2m} . \]
\end{enumerate}

Finally, we discuss how to encode Corollary~\ref{cor:parity}.
The simplest case of Corollary~\ref{cor:parity} tells us that an even number of $\{b_m,c_m,d_m\}$
are true.  This is encoded as the four clauses
\begin{gather*}
b_m \lor c_m \lor d_m , \quad
\lnot b_m \lor \lnot c_m \lor d_m , \\
\lnot b_m \lor c_m \lor \lnot d_m , \quad
b_m \lor \lnot c_m \lor \lnot d_m .
\end{gather*}
The general case of Corollary~\ref{cor:parity}
is an XOR constraint of length~$5$.
While SAT solvers are not good with arbitrary XOR constraints
we found that these constraints were short enough that they could be effectively encoded
as the sixteen clauses with an even number of negated literals, i.e.,
\[ a_k\lor a_{2k}\lor b_k\lor c_k\lor d_k, \dotsc, a_k\lor \lnot a_{2k}\lor \lnot b_k\lor \lnot c_k\lor \lnot d_k . \]

\subsection{Postprocessing: Remove equivalent good matrices}

After all the SAT instances generated in step~4 are solved we will have
a list of all circulant good matrices of order~$n$ up to equivalence.  However, some of
them may be equivalent to each other since we have not encoded all of the
equivalence operations of Section~\ref{sec:equiv}.  In order
to provide a list which contains only inequivalent good matrices
it is necessary to test the matrices for equivalence and remove
any that are found to be equivalent to one which has previously
been found.

To do this, we define a canonical representative of each equivalence
class of defining rows of good matrices.
This way, to check if two good matrices are
equivalent we compute the canonical representative of each and check
if they are equal.

First, consider the first two equivalence operations
(reorder and negate).  Given a quadruple of defining rows $(A,B,C,D)$
we apply the negation operation
to set the first entry of each defining row to be positive.  Then
we apply the reorder operation to sort $B$, $C$, $D$ in lexicographically
increasing order.  This gives us a canonical representative of each equivalence
class formed by the first two equivalence operations which we denote
by $M_{A,B,C,D}$.

Next, consider the third equivalence operation (permute indices).
Let $\sigma$ be an automorphism of $\Z_n$ and let $\sigma$ act on
$(A,B,C,D)$ by permuting the indices of the defining rows.
This operation commutes with the first two operations, so the lexicographic
minimum of the set
\[ S \coloneqq \{ \, M_{\sigma(A,B,C,D)} : \sigma \in \Aut(\Z_n) \, \} \]
is the lexicographic minimum of all quadruples equivalent to $(A,B,C,D)$
and is therefore canonical.  We compute this by explicitly generating
all elements of $S$ and selecting the lexicographic minimum.

\subsection{Optimizations}

For the benefit of those who would like to implement our algorithm we
now discuss some optimizations that we found useful.

In the programmatic encoding and in step~2 of our enumeration
algorithm we need to check if there exists an integer $k$ such that
$\PSD_X(k)>4n$ for a given sequence $X$.  We only need to check $k$
with $0\leq k<n/2$ because the symmetries of the Fourier
transform on real input $X$ imply that $\PSD_X(k)=\PSD_X(\pm k\bmod n)$.
In fact, checking $k=0$ is also unnecessary because $\PSD_X(0)=\row(X)^2$
and step~1 will filter any sequences $X$ with $\row(X)^2>4n$.  Also,
when checking the PSD criterion for each $k$ we sort the elements of $S$ so that
the sequences with the largest PSD values appear first, making it easier to learn short clauses.

In step~3, the matching procedure can be done very efficiently by a string
sorting technique similar to the method used in~\cite{kotsireas2010efficient}.
First we enumerate all $(A',B')\in\Ssk\times\Ssy$ that satisfy
$\PSD_{A'}(k)+\PSD_{B'}(k)\leq 4n$ for all $k$ and all $(C',D')\in\Ssy^2$ that satisfy
$\PSD_{C'}(k)+\PSD_{D'}(k)\leq 4n$ for all $k$.
At this point we form one list that contains strings that contain the values
of $\PSD_{A'}(k)+\PSD_{B'}(k)$ and another list that contains strings that contain
the values of $4n-(\PSD_{C'}(k)+\PSD_{D'}(k))$
and look for strings which occur in both lists.
This can be done by sorting the two lists and doing a linear scan through
the lists to find duplicates.
The fact that PSDs are real numbers can make this tricky, but we can deal
solely with integers by applying the inverse Fourier transform to the PSD values;
this produces integer values known as PAF values and the relationship given in
Theorem~\ref{thm:compress} becomes (for $0<k<n$)
\[ \PAF_{A'}(k)+\PAF_{B'}(k)+\PAF_{C'}(k)+\PAF_{D'}(k) = 0 .
\]
In fact, since this step was not a bottleneck
we computed the PAF values using the slower but
more straightforward relationship $\PAF_X(k)=\sum_{j=0}^{n-1}x_jx_{j+k\bmod n}$.

In step~4, before calling the SAT solver we partition the compressed
quadruples into equivalence classes using the reorder and automorphism
equivalence operations from
Section~\ref{sec:background} and keep one quadruple from
each equivalence class.  If any satisfiable SAT instances are removed
by this process they will have equivalent solutions (also with positive first entries)
in the instance that was kept.

\section{Results}\label{sec:results}

In this section we discuss our implementation of our enumeration
algorithm, the software and hardware we used to run it,
and the results that we achieved.  Our code is available
from our website
\href{https://uwaterloo.ca/mathcheck/}{\nolinkurl{uwaterloo.ca/mathcheck}}.

Step~1 of our algorithm was done using \textsc{nsoks} in \textsc{Maple}~\cite{nsoks}.
Steps~2 and~3 were done with custom C++ code and the string sorting was done
with the GNU core utility \textsc{sort}.  Step~4 was done with the programmatic
SAT solver \textsc{MapleSAT}~\cite{aaai2016} and the PSD values were computed
using the library \textsc{FFTW}~\cite{frigo2005design}.  Because \textsc{FFTW}
uses floating point values which are inherently inaccurate all comparisons
were checked to a tolerance of $\epsilon=10^{-2}$ which is small but larger
than the precision of the FFT used.
The computations were run on a cluster of 64-bit Opteron 2.2GHz and Xeon 2.6GHz processors
running \mbox{CentOS}~6.9 and using 500MB of memory.
By far the most computationally expensive step of our algorithm was running the SAT solver.
This accounted for 99.9\% of the total running time,
underscoring the importance of using a state-of-the-art SAT solver such
as \textsc{MapleSAT}.
This step was parallelized across 250 cores and
completed in about two weeks.
Despite step~4 being the bottleneck the other steps were found to be
essential could not be removed
without significantly increasing the computation time.

The timings for running our implementation 
in each odd order divisible by $3$
up to order $70$ are given in Table~\ref{tbl:results}.
This table also contains the number of SAT instances generated in each order
and the number of inequivalent good matrices found (denoted by $\#G_n$).
\begin{table}
\begin{center}
\begin{tabular}{c@{\qquad}c@{\qquad}c@{\qquad}c}
$n$        & Time (h)   & \# inst.   & $\#G_n$    \\ 
3          & 0.00       & 1          & 1          \\ 
9          & 0.00       & 2          & 1          \\ 
15         & 0.00       & 11         & 11         \\ 
21         & 0.00       & 39         & 10         \\ 
27         & 0.00       & 186        & 13         \\ 
33         & 0.01       & 840        & 15         \\ 
39         & 0.07       & 1934       & 5          \\ 
45         & 1.91       & 19205      & 4          \\ 
51         & 11.35      & 23611      & 0          \\ 
57         & 233.56     & 102402     & 1          \\ 
63         & 11115.70   & 808642     & 0          \\ 
69         & 72366.85   & 918940     & 0          
\end{tabular}
\end{center}
\caption{The running time in hours, number
of SAT instances used, and number of inequivalent
good matrices found ($\#G_n$) in each odd order $n<70$ divisible by $3$.}\label{tbl:results}
\end{table}
Our exhaustive search in order~$27$ produced 13 sets of inequivalent good matrices, though
Szekeres' classic paper~\cite{szekeres1988note} contains only 12 sets of such matrices
and Seberry reports that this was checked by at least one other researcher~\cite{seberry1999}.
Comparing the matrices generated using our method with those in his paper shows that
Szekeres missed a set of good matrices.
Furthermore, we found a new set of good matrices of order $57$;
the two new sets of good matrices produced by our method are given in Figure~\ref{fig:newgood}.

\newcommand\p{$\texttt{+}$}
\newcommand\m{$\texttt{-}$}
\begin{figure*}
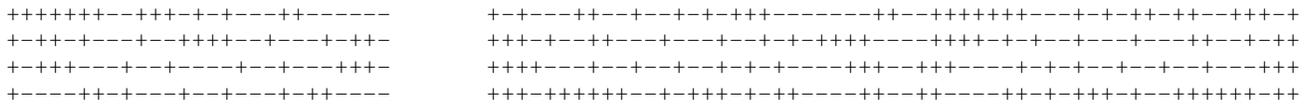
\centering\small
\p\p\p\p\p\p\p\m\m\p\p\p\m\p\m\p\m\m\m\p\p\m\m\m\m\m\m \qquad\qquad \p\m\p\m\m\m\p\p\m\m\p\m\m\p\m\p\m\p\p\p\m\m\m\m\m\m\m\p\p\m\m\p\p\p\p\p\p\p\m\m\m\p\m\p\m\p\p\m\p\p\m\m\p\p\p\m\p \\
\p\m\p\p\m\p\m\m\m\p\m\m\p\p\p\p\m\m\p\m\m\m\p\m\p\p\m \qquad\qquad \p\p\p\m\p\m\m\p\p\m\m\m\p\m\m\m\p\m\m\p\m\p\m\p\p\p\p\m\m\m\m\p\p\p\p\m\p\m\p\m\m\p\m\m\m\p\m\m\m\p\p\m\m\p\m\p\p \\
\p\m\p\p\p\m\m\m\p\m\m\p\m\m\m\m\p\m\m\p\m\m\m\p\p\p\m \qquad\qquad \p\p\p\p\m\m\m\p\m\m\p\m\m\p\m\m\p\m\p\m\p\m\m\m\m\p\p\p\m\m\p\p\p\m\m\m\m\p\m\p\m\p\m\m\p\m\m\p\m\m\p\m\m\m\p\p\p \\
\p\m\m\m\m\p\p\m\p\m\m\m\p\m\m\p\m\m\m\p\m\p\p\m\m\m\m \qquad\qquad \p\p\p\m\p\p\p\p\p\p\m\m\p\m\p\p\p\m\p\m\p\p\m\m\m\m\p\p\m\m\p\p\m\m\m\m\p\p\m\p\m\p\p\p\m\p\m\m\p\p\p\p\p\p\m\p\p
\caption{The defining rows of new good matrices of order 27 and order 57
where {\p} encodes $1$ and {\m} encodes $-1$.}\label{fig:newgood}
\end{figure*}

\section{Conclusion}\label{sec:conclusion}

In this paper we have demonstrated the applicability of the SAT+CAS
and programmatic SAT paradigms to the problem of finding good matrices
in combinatorics.  This problem has been well-studied for almost 50 years
and the algorithm we've developed using SAT solvers coupled with
computer algebra systems and mathematical libraries is a
practical method of enumerating circulant good matrices
(at least those with orders divisible by $3$).
This is evidenced by the large gap between the prior state-of-the-art
and the results of this paper---%
prior to this year the largest order enumerated was $39$
and the orders up to $49$ were only enumerated this year, while
our algorithm was able to enumerate circulant good matrices of order $69$.
Since the search space grows exponentially in the order
(approximately like $4^n$) this is a much larger
search space than the order $49$ search space.

Additionally, we've demonstrated the effectiveness of our algorithm by
constructing two new sets of good matrices, including one that escaped
detection in a 1988 search~\cite{szekeres1988note}
and a double-check by Koukouvinos in 1999~\cite{seberry1999}.  It is surprising that
this set of good matrices was overlooked for so long though the fact that the
double-check also failed to produce it likely dissuaded researchers
from spending more computational resources to triple-check the
same search in order~$27$.

Unfortunately, the fact that there are no known certificates for
an exhaustive search means that it is difficult to verify
that a search completed successfully.  In our case, we relied on
multiple pieces of software including
\textsc{Maple}, \textsc{MapleSAT}, \textsc{FFTW},
\textsc{GNU sort}, and custom-written C++ code.
A bug in any one of these programs could cause a
good matrix to be overlooked.
The fact that our code confirms (and even corrects) the results of
previous searches gives some assurance that it is working as intended.
However, verification of our nonexistence results from multiple
independent sources would be ideal.

\section*{Acknowledgments}

We thank the reviewers for their detailed comments which
improved this paper's clarity.  The computations
were made possible by the high-performance computer clusters
administered by Compute Canada and SHARCNET.

\bibliography{good}
\bibliographystyle{aaai}

\section*{Appendix: Proofs}

\renewcommand{\thetheorem}{\ref{thm:psd}}
\begin{theorem}
Let $A$, $B$, $C$, $D$ be sequences of length $n$
with $A$ skew and $B$, $C$, $D$ symmetric.
Then $A$, $B$, $C$, $D$ are the defining rows of a set
of good matrices if and only if
for all integers $k$ we have 
\[ \PSD_A(k) + \PSD_B(k) + \PSD_C(k) + \PSD_D(k) = 4n . \]
Here $\PSD_X(k) \coloneqq \abs[\big]{\sum_{j=0}^{n-1}x_j\omega^{jk}}^2$
with $\omega\coloneqq\exp(2\pi i/n)$ is
the power spectral density of $X=[x_0,\dotsc,x_{n-1}]$.
\end{theorem}
\begin{proof}
Using the notation of~\cite{djokovic2015compression}
an alternative way of stating condition~(c) in our
definition of good matrices is that
\[ \PAF_A + \PAF_B + \PAF_C + \PAF_D = [4n,0,\dotsc,0] . \]
Then \cite[Theorem~2]{djokovic2015compression} gives the desired result.
\end{proof}

\renewcommand{\thetheorem}{\ref{cor:psd}}
\begin{corollary}
Let $A$, $B$, $C$, $D$ be the defining rows of a set of good matrices
of order~$n$.
If\/ $S$ is a subset of\/ $\{A,B,C,D\}$ then
\[ \sum_{X\in S} \PSD_X(k) \leq 4n \]
for all integers $k$.
\end{corollary}
\begin{proof}
By Theorem~\ref{thm:psd} we have
\[ \sum_{X=A,B,C,D} \PSD_X(k) = 4n . \]
The desired inequality follows from the fact that PSD values
are nonnegative.
\end{proof}

\renewcommand{\thetheorem}{\ref{cor:dio}}
\begin{corollary}
Let $A$, $B$, $C$, $D$ be the defining rows of a set of good matrices
of odd order~$n$.
Then
\[ \row(B)^2 + \row(C)^2 + \row(D)^2 = 4n-1 \]
where $\row(X)$ denotes the sum of entries in $X$.
\end{corollary}
\begin{proof}
Since $\PSD_X(0)=\row(X)^2$ and $\row(A)=1$ (as $A$ is skew) the desired
equality follows by setting $k=0$ in Theorem~\ref{thm:psd}.
\end{proof}

\renewcommand{\thetheorem}{\ref{lem:sign}}
\begin{lemma}
Let $(a_0,\dotsc,a_{n-1})$, $(b_0,\dotsc,b_{n-1})$, $(c_0,\dotsc,c_{n-1})$, $(d_0,\dotsc,d_{n-1})$ be the defining rows
of a set of good matrices $A$, $B$, $C$, $D$ of odd order $n$ with $b_0=c_0=d_0=1$.  Then
\[ \row(B) \equiv \row(C) \equiv \row(D) \equiv n \pmod{4} . \]
\end{lemma}
\begin{proof}
When $X$ is symmetric and $x_0=1$ we have
\[ \row(X)=1+2(x_1+\dotsb+x_{(n-1)/2}) . \]
Since each $x_k$ is $\pm1$ and congruent to $1\pmod{2}$ the
sum in parenthesis above is congruent to $(n-1)/2\pmod{2}$
and twice the sum is congruent to $n-1\pmod{4}$.
Thus $\row(X)\equiv n\pmod{4}$.
\end{proof}

\renewcommand{\thetheorem}{\ref{thm:product}}
\begin{theorem}
Let $(a_0,\dotsc,a_{n-1})$, $(b_0,\dotsc,b_{n-1})$, $(c_0,\dotsc,c_{n-1})$, $(d_0,\dotsc,d_{n-1})$ be the defining rows
of a set of good matrices $A$, $B$, $C$, $D$ of odd order $n$.  Then $a_k b_k c_k d_k = -a_{2k\bmod n}b_0c_0d_0$ for all\/ $1\leq k<n$.
\end{theorem}
\begin{proof}
We may assume that $b_0=c_0=d_0=1$; all sets of good matrices are equivalent to one in which this holds
and the once the theorem has been proven for this class of good matrices the general theorem follows by
applying the negation equivalence operation appropriately.

Let $P_X$ be the matrix whose positive entries are identical to the entries of $X$ and whose entries are $0$ otherwise.
Since $X\in\{A,B,C,D\}$ has $\pm1$ entries we have $P_X = (X+J)/2$ where $J$ is a matrix of all $1$s and $P_X=\sum_{x_i=1}T^i$ where $T$ is the circulant matrix with first row $(0,1,0,\dotsc,0)$.  Let $p_X$ be the number of positive entries in the first row of $X$.  In~\cite[\S3.4]{gene2015} it is shown that
we have
\[ -P_A^2 + P_A + P_B^2 + P_C^2 + P_D^2 = -(n-p_B-p_C-p_D)J + nI . \]
Since $n$ is odd and $p_B$, $p_C$, and $p_D$ are odd by the symmetry of $B$, $C$, and~$D$, we have that
the right side reduces to $I\pmod{2}$.  Since $P_X^2\equiv\sum_{x_i=1}T^{2i}\pmod{2}$ the left side
mod $2$ reduces to
\[ \sum_{a_i=1}(T^{2i}+T^i) + \sum_{b_i=1}T^{2i} + \sum_{c_i=1}T^{2i} + \sum_{d_i=1}T^{2i} . \]
Let $i/2$ denote the unique integer in $\{0,\dotsc,n-1\}$ such that $2(i/2)\equiv i\pmod{n}$ which is defined since $n$ is odd.
Putting together the above, we derive (mod $2$)
\[ \sum_{a_{i/2}=1}T^i + \sum_{a_i=1}T^i + \sum_{b_{i/2}=1}T^i + \sum_{c_{i/2}=1}T^i + \sum_{d_{i/2}=1}T^i \equiv I . \]
Consider the $i$th entry of the initial row of the matrix on the left where $1\leq i<n$.
This entry will be the sum of the $1$s in the set $\{a_{i/2},a_i,b_{i/2},c_{i/2},d_{i/2}\}$.  The right side says that this
entry is $0\pmod{2}$ and therefore an even number of entries in this set are $1$ (and an odd number must be $-1$).
In other words,
\[ a_{i/2}a_ib_{i/2}c_{i/2}d_{i/2} = -1 \]
which a rewriting of the required result with the definition $k=i/2$.
\end{proof}

Let $\bar{x}\coloneqq(1-x)/2$, i.e., the mapping $x\mapsto\bar{x}$ is the group
isomorphism between $\Z_4^*=\{\pm1\}$ and $\Z_2=\{0,1\}$.

\renewcommand{\thetheorem}{\ref{cor:parity}}
\begin{corollary}
Let $(a_0,\dotsc,a_{n-1})$, $(b_0,\dotsc,b_{n-1})$, $(c_0,\dotsc,c_{n-1})$, $(d_0,\dotsc,d_{n-1})$ be the defining rows
of a set of good matrices $A$, $B$, $C$, $D$ of odd order $n$ with
$b_0=c_0=d_0=1$.
Then for $1\leq k<n/2$ we have in\/ $\Z_2$
\[ \bar{a}_k + \bar{a}_{2k} + \bar{b}_k + \bar{c}_k + \bar{d}_k = 1 \]
and when $n=3m$ we have $\bar{b}_{m}+\bar{c}_{m}+\bar{d}_{m}=1$.
\end{corollary}
\begin{proof}
Applying the mapping $x\mapsto\bar{x}$ to $a_kb_kc_kd_k=-a_{2k}$
from Theorem~\ref{thm:product} gives
\[ \bar{a}_k + \bar{b}_k + \bar{c}_k + \bar{d}_k = 1 + \bar{a}_{2k} \]
which implies the first result since $-1=1$ in $\Z_2$.
When $k=n/3$ we have $\bar{a}_{2k}=-\bar{a}_{n-2k}=-\bar{a}_k$
and the second result follows.
\end{proof}

\end{document}